\documentclass[12pt]{article}
\usepackage{amsfonts,amsmath}
\usepackage[mathscr]{eucal}
\usepackage{amssymb}
\usepackage{amsthm}
\theoremstyle{plain}
\newtheorem{thm}{Theorem}

\def\theequation{\arabic{section}.\arabic{equation}}
\newcommand{\be}{\begin{eqnarray}}
\newcommand{\ee}{\end{eqnarray}}
\newcommand{\nn}{\nonumber \\}
\newcommand{\lb}{\label}
\newcommand{\p}[1]{(\ref{#1})}

\begin{document}

\begin{titlepage}

\vspace*{0.2cm}

\renewcommand{\thefootnote}{\star}
\begin{center}

{\LARGE\bf  Comments on HKT  supersymmetric sigma  }\\

\vspace{0.5cm}

{\LARGE\bf models and their Hamiltonian reduction }\\

\vspace{1.5cm}
\renewcommand{\thefootnote}{$\star$}

{\large\bf Sergey~Fedoruk} ${}^{a\,b\,\ddagger}$,
\quad {\large\bf Andrei~Smilga} ${}^c$
 \vspace{0.5cm}

{${}^a$ \it Bogoliubov Laboratory of Theoretical Physics, JINR,}\\
{\it 141980 Dubna, Moscow region, Russia} \\
\vspace{0.1cm}

{\tt fedoruk@theor.jinr.ru}\\
\vspace{0.7cm}

{${}^b$ \it Dept. Theor. Phys., Tomsk State Pedagogical University, \\
Tomsk, 634061, Russia}  \\

\vspace{0.7cm}

{${}^c$\it SUBATECH, Universit\'e de Nantes,}\\
{\it 4 rue Alfred Kastler, BP 20722, Nantes 44307, France;}\\
\vspace{0.1cm}

{\tt smilga@subatech.in2p3.fr}\\

\end{center}
\vspace{0.2cm} \vskip 0.6truecm \nopagebreak

\begin{abstract}
\noindent
Using complex notation, we present new simple expressions for two pairs
of complex supercharges in HKT supersymmetric sigma models. The second pair
of supercharges depends on the holomorphic antisymmetric
``hypercomplex structure'' tensor ${\cal I}_{jk}$ which plays the same role
for the HKT models as the complex structure tensor for the K\"ahler models. When
the Hamiltonian and supercharges commute with the momenta conjugate to the
imaginary parts of the complex coordinates, one can perform a Hamiltonian reduction. The models thus
obtained represent a special class of quasicomplex sigma
models introduced recently in \cite{quasi}.

\end{abstract}

\vspace{1cm}

\phantom{Keywords: }

\vspace{1cm}

-----------------------

\vspace{0.3cm}

{\footnotesize $^{\ddagger}$ On leave of absence from V.N.\,Karazin Kharkov National University, Ukraine.}

\newpage

\end{titlepage}

\setcounter{footnote}{0}

\setcounter{equation}0
\section{Introduction}
Supersymmetric quantum mechanical (SQM) models describing the motion
of a supersymmetric particle on a curved manifold have been studied since
\cite{Witten}. Most of these problems represent a reformulation of
classical problems of differential geometry. In particular, the model analyzed
in \cite{Witten} boils down to the well-known de Rham complex.

 The powerful supersymmetry formalism allows one to reproduce
 known mathematical results in a simple way.
 In this regard,  one can mention the famous Atiyah-Singer theorem \cite{AS}.
A pure mathematical proof of this theorem is rather complicated. On the other hand, its supersymmetric
proof using the functional integral formalism \cite{ASsusy} (see also \cite{IS,HRR}) is transparent
and beautiful.

 But supersymmetry makes it also possible to derive {\it new} results. In particular, it allows to
construct new differential geometry structures not studied before by mathematicians. For example,
the SQM model studied in \cite{Witten} and involving an extra potential is called now
``Witten deformation of the de Rham complex''. There are other deformations of the classical
de Rham and Dolbeault complexes involving torsions \cite{torsion,FIS1}.
The HKT models (the subject of the present paper) were first introduced
by physicists in the supersymmetric sigma model framework \cite{HKT} (see also earlier papers
\cite{IvKrLev,SSTvP,DelVal,Papa95}
where some elements of the HKT structure were displayed) and only then were described
in pure mathematical terms \cite{Grant,Verb}. Less known CKT and OKT
models \cite{CKT,Hull,FIS2} are still awaiting their appreciation by mathematicians.
The same concerns the recently discovered quasicomplex sigma models.

To find a way in this multitude of models, one needs road maps.
We noticed in \cite{taming} that {\it all} these models can be obtained from
the trivial flat Dolbeault model with
  \be
 \lb{triv}
Q = \psi_a \pi_a , \ \ \ \ \ \ \ \ \bar Q = \bar \psi_a \bar \pi_a, \ \ \ \ \
 \ H = \bar \pi_a \pi_a
  \ee
by two operations: {\it (i)} similarity transformation of complex supercharges
and {\it (ii)} Hamiltonian reduction. In particular, a similarity
transformation
 \be
\lb{sim}
Q \ \to \ e^R Q e^{-R}, \ \ \ \ \ \bar Q \ \to \ e^{-R^\dagger} \bar Q
e^{R^\dagger}
  \ee
with $R = \omega_{ab}  \psi_a \bar \psi_b $ applied to \p{triv} gives a model
describing a nontrivial Dolbeault complex. If the metric
 \be
\lb{metrhjk}
h_{m\bar{n}} \ =\ \left(e^{-\omega} e^{-\omega^\dagger} \right)_{m \bar n}
 \ee
thus obtained does not depend on imaginary parts of the complex coordinates
$z^m$, the momenta $\pi_m - \bar \pi_m$ commute with the Hamiltonian and one
can perform a Hamiltonian reduction giving a model with half as much bosonic
degrees of freedom $\{{\rm Re} (z^m) \}$
\footnote{The Hamiltonian can, of course, commute with any number of momenta.
The corresponding Hamiltonian reductions give different models some of
which were discussed in \cite{taming}. In this paper, we will discuss only the
Hamiltonian reduction with respect to all imaginary parts of $z^m$. }.
 If the Hermitian metric \p{metrhjk} involves an imaginary part,
\be
\lb{sym+antisym}
h_{m {\bar n}} \ =\ \frac 12 \left(
g_{(m {\bar n})} +  ib_{[m {\bar n}]} \right) \, ,
 \ee
we obtain a quasicomplex model \cite{quasi} (the origin of the factor $1/2$
in \p{sym+antisym} will be clarified later). If $b_{[m {\bar n}]} = 0$, we
obtain a usual de Rham model of \cite{Witten}.

Both Dolbeault and de Rham models can have extended supersymmetries.
The de Rham model with an extra pair of supercharges can be formulated
for K\"ahler even-dimensional manifolds \cite{Zumino,DMPH,MP}. Mathematicians know
this model as the K\"ahler -- de Rham complex. There are also ${\cal N} = 8$ supersymmetric
(i.e. including 8 different real supercharges)
de Rham models with 3 extra pairs of supercharges and defined on hyper-K\"ahler manifolds.
A Dolbeault model with an extra pair of supercharges is called an HKT model
\footnote{ HKT stands for {\it hyper-K\"ahler with torsion}. This name is probably a little
bit misleading because these manifolds are {\it not} hyper-K\"ahler and not even K\"ahler, but
a better one was not invented.}. If its metric does not depend on ${\rm Im}
(z^m) $, one can perform a Hamiltonian reduction.

Our main observation is that a model thus obtained belongs to the class
of quasicomplex models representing their special type. It enjoys ${\cal N} = 4$ supersymmetry.

 The explicit component expressions for the HKT supercharges
were derived in \cite{QHKT}. However, they were written in terms of real coordinates. To perform
the Hamiltonian reduction described above, we need first to represent them in complex form.
If expressing in proper terms, the corresponding expressions turn out
to be very simple [see Eq.\p{SR-HKT-s} below]. This representation make manifest the
kinship between the mathematical
structure of the HKT models and the structure of K\"ahler -- de Rham models. The latter are
characterized by a presence of  the closed K\"ahler form. The components of this form define
the complex structure tensor $I_{MN}$. Similarly, an HKT manifold is characterized by the presence
of a closed holomorphic $(2,0)$ - form. Its components define a holomorphic tensor ${\cal I}_{mn}$ which
may be called a {\it hypercomplex structure} tensor.
\footnote{Throughout the paper, the real tensor indices  are
denoted by large latin letters $M,N,\ldots$, while small latin letters $m, \bar m, \ldots$
are reserved for the holomorphic and antiholomorphic complex indices.}

 The plan of the paper is the following. Sect. 2 represents a mathematical introduction where
 we {\it translate} many facts known to mathematicians into a language understandable to physicists.
\footnote{Unfortunately, the papers written by mathematicians and by theorists doing mathematical
physics are written in rather different languages, even when they are devoted to basically
the same subject. More often than not they are mutually not understandable and translation is
necessary.} In Sect. 3, after reminding how simple expressions for the supercharges can be derived
in ${\cal N} = 2$ models (the main idea is to treat the fermions with world
indices rather than the fermions with tangent
space indices as basic dynamical variables), we present new nice
 generic expressions for the complex HKT supercharges as well as
 the supercharges obtained after their Hamiltonian reduction.

In Sect. 4 (the central section of the paper),
we discuss the Hamiltonian reduction procedure invoking superfield formalism.
A generic Dolbeault ${\cal N} = 2$ model is expressed via ({\bf 2},\,{\bf 2},\,{\bf 0}) chiral superfields
\footnote{We follow the notation of \cite{PT} such that the numerals count the numbers of the physical bosonic,
physical fermionic and auxiliary bosonic fields.}.
When the metric depends only on real parts of the coordinates, one can perform
the Hamiltonian reduction with respect to imaginary parts. The reduced model is described in terms of
 ({\bf 1},\,{\bf 2},\,{\bf 1}) multiplets --- the imaginary parts of the coordinates are traded for
 auxiliary fields. Likewise, ${\cal N} = 4$ HKT models are described by  ({\bf 4},\,{\bf 4},\,{\bf 0}) multiplets
that involve four real or two complex coordinates.
After reduction, imaginary parts of the latter are traded for auxiliary fields and we are led to
({\bf 2},\,{\bf 4},\,{\bf 2}). Generically, one obtained a {\it deformed} K\"ahler -- de Rham complex
which involves extra ``quasicomplex'' terms. At the superfield level, such models involve, besides the familiar
K\"ahler potential term, a holomorphic $F$-term of some special form  [see Eq. \p{act242}].

 We emphasize that this type of  Hamiltonian reduction differs from the Hamiltonian reduction for
hyper-K\"ahler manifolds
\cite{HiKLR,GiRG,SmiS} and HKT manifolds \cite{GraPP} studied earlier. In \cite{HiKLR,GiRG,SmiS,GraPP},
the reduction related models of the same type: hyper-K\"ahlerian models to hyper-K\"ahlerian and HKT to HKT.
In our case, the reduction changes the geometry: a Dolbeault model gives after reduction a quasicomplex de Rham model and
an HKT model gives a quasicomplex K\"ahler model.

Short conclusions are drawn in the last section.

In Appendix\,A, we discuss in details how {\it Hamiltonian} reduction is  described
in {\it Lagrangian } component formalism.
In Appendix\,B,  we present complete component
Lagrangians of the original HKT theory with several interacting ({\bf 4},\,{\bf 4},\,{\bf 0}) multiplets and of
the quasicomplex K\"ahler -- de Rham theory with
 several interacting ({\bf 2},\,{\bf 4},\,{\bf 2}) multiplets.
In Appendix\,C, we give some technical details concerning establishing the correspondence between a
generic HKT model admitting reduction and its reduced K\"ahler -- de Rham quasicomplex daughter.

\section{Two definitions of HKT manifolds and their equivalence.}
\setcounter{equation}0

We assume that the reader is familiar with the geometry of K\"ahler and hyper-K\"ahler manifolds.
For a reader physicist, we can recommend the excellent review \cite{revcomplex}. The basic facts are
the following:
 \begin{itemize}
\item A {\it K\"ahler manifold} is characterized by an antisymmetric complex structure tensor
$I_{MN}$.~\footnote{Our index policy is the following. {\it (i)} Capital latin letters denote the indices in ${\cal R}^N$.
{\it (ii)} small latin letters are reserved for the indices of holomorphic variables. {\it (iii)}
In most cases but not always,
the indices of antiholomorphic variables are marked with a bar ($\bar z^{\bar m}$ etc) . {\it (iv)}
By the reasons which become clear later, the holomorphic indices in Sect.4.2.3 are Greek.}
The property $I_{MN} I^{NK} = - \delta_M^K$ holds. $I_{MN}$ is covariantly constant,
$\nabla_P I_{MN} = 0$. It follows that the K\"ahler form $\Omega = I_{MN} \, dx^M \wedge dx^N$ is closed,
$d\Omega = 0$.

  \item A generic complex manifold also involves an antisymmetric complex structure tensor $I$, but
 the standard covariant derivative
$\nabla_P I_{MN}$ (with symmetric Christoffel symbols) does not necessarily vanish. $I$ should satisfy,
however, certain integrability conditions,
  \be
\lb{Nijen}
 \nabla_{[M} I_{N] P} \ =\ I_M^{\ Q} I_N^{\ S} \nabla_{[Q} I_{S] P} \, .
 \ee
Eq.\p{Nijen} amounts to the vanishing of the so called Nijenhuis tensor.
\footnote{The Nijenhuis tensor is defined as
\be
\label{Nijenhuis}
N^I_{JK} = I^M_{\ [J} \partial_M I^I_{\ K]} \ -\ I^I_{\ M} \partial_{[J} I^M_{\ K]} \, .
 \ee
Its vanishing may be expressed as a condition
 \be
\label{NijenbezGamm}
\partial_{[M} I_{N]}^{\ P} \ = \  I_M^{\ Q} I_N^{\ S} \partial_{[Q} I_{S]}^{\ P} \, .
 \ee
One can observe that one can as well replace the usual derivatives in \p{NijenbezGamm} by covariant ones.
Lowering the index $P$ then gives \p{Nijen}.
}
 It is necessary to be able
to define  (anti)holomorphic
coordinates $x^M = \{z^m, \bar z^{\bar m} \}$ with Hermitian metric, $ds^2 = 2h_{m \bar n} dz^m d{\bar z}^{\bar n}$ on the whole manifold.
In addition, if \p{Nijen} does not hold, nilpotent supercharges cannot be constructed.

When the complex coordinates  are chosen, the tensor $I_M^{\ N}$ has the following
nonzero components,
  \be
\lb{Icompl}
I_m^{\ n} = -I^n_{\ m} = -i \delta_m^n, \ \ \ I_{\bar m}^{\ \bar n} = -I^{\bar n}_{\ \bar m} = i
\delta_{\bar m}^{\bar n} \, .
  \ee
It follows that $I_{m {\bar n}} = - I_{{\bar n} m} = -ih_{m{\bar n}}$.

\item As was mentioned, a standard covariant derivative of $I_{MN}$ does not
generically vanish. However,
for any $I$ satisfying the conditions above, one can define an affine connection
   \be
\lb{affineconn}
 \hat \Gamma^M_{NK} \ =\  \Gamma^M_{NK} + \frac 12 g^{ML} C_{LNK}
 \ee
with the torsion tensor $C_{LNK}$ antisymmetric under $N \leftrightarrow K$
such that $\hat \nabla_P I_{MN} = 0$. If one requires for the tensor   $C_{LNK}$ to be
totally antisymmetric, such connection is unique and is called {\it Bismut connection}
\cite{Bismut}. Explicitly,
\be
\lb{CBismut}
 C^{(Bismut)}_{MNK}(I) \ =\  I_M^{\ P} I_N^{\ Q} I_K^{\ R} (\nabla_P I_{QR} +  \nabla_Q I_{RP} +
\nabla_R I_{PQ}) \, .
\ee
In complex coordinates, this tensor involves only the components of the type
$(2,1)$ and $(1,2)$. The explicit expressions  are \cite{IS}
\be
\lb{Ccompl}
C_{mn{\bar p}}  =  C_{n{\bar p}m} \ = C_{{\bar p}mn} = \ \partial_n h_{m \bar p} - \partial_m h_{n {\bar p}}, \nn
C_{{\bar m}{\bar n}p} =  C_{{\bar n}p {\bar m }} = C_{p{\bar m}{\bar n}} =
\ \partial_{\bar n} h_{{\bar m} p} - \partial_{\bar m} h_{{\bar n} p} \, .
\ee

\item A  {\it hyper-K\"ahler manifold} has three different antisymmetric
covariantly constant complex structures $I,J,K$
satisfying the quaternion algebra
\be
\lb{quaternion}
I^2 = J^2 = K^2 = -1, \ \ \ \ \
 IJ = K, \ \ \  JK = I, \ \ \  KI = J \, .
 \ee

\item Finally, we define a {\it hypercomplex manifold} as a manifold with three integrable quaternionic complex structures
whose standard  covariant derivatives do not necessarily vanish.
The real dimension of a hypercomplex manifold
is an integer multiple of 4 --- the same as for the hyper-K\"ahler manifolds.

  \end{itemize}

We go over now to the HKT manifolds. There are two equivalent  definitions:

\vspace{.1cm}

{\bf Definition 1.}
An HKT manifold is a hypercomplex manifold where the complex structures satisfy
 an additional constraint: they  are covariantly constant with {\it one and the same}
 torsionful Bismut affine connection,
 \be
\lb{CICJCK}
  C_{MNK}(I) \ =\  C_{MNK}(J) \ =\  C_{MNK}(K) \, .
 \ee

\vspace{.1cm}

{\bf Definition 2.} An HKT manifold is a hypercomplex manifold  where  the $(2,0)$ - form
\be
\lb{omegahol}
\omega = \ \Omega_J + i\Omega_K = \ (J + iK)_{MN} \, dx^M \wedge dx^N
 \ee
(we will shortly see that it is  holomorphic with respect to $I$) is closed,
$$\partial_I \omega = 0$$.

\vspace{.1cm}

We will give a proof here for the half of the equivalence theorem (see e.g. \cite{Verb} for another half).
Taking \p{CICJCK} as a basic definition (suggested originally in \cite{HKT}), we construct the closed
holomorphic $(2,0)$ - form. The existence of such form was first proven in \cite{Gauduchon}.
We follow here much more user-physicist-friendly \cite{Grant}.

As a first step, we introduce two  operators  associated with the complex structure $I$ and acting
on $n$ -- forms. The operator  $\iota$ is defined  according to
  \be
\lb{defiota}
&&{\rm if} \ \ \ \ \ \ \ \ \ \ \omega \ =\   \omega_{M_1 \ldots M_n} dx^{M_1} \wedge \cdots \wedge
dx^{M_n}\,, \nn
&&{\rm then} \ \ \ \ \iota \omega \ =\ n \, \omega_{N [M_2 \ldots M_n} (I)^{N}_{\ M_1]}
 dx^{M_1} \wedge \cdots \wedge dx^{M_n}\,.
 \ee
For a form $\omega_{p,q}$ with $p$ holomorphic and $q$ antiholomorphic indices,
  \be
\label{iotapq}
\iota \omega_{p,q} \ =\ i(p-q) \omega_{p,q} \, .
 \ee

Another operator $\omega \to I\omega$ is defined as
   \be
\lb{defIomega}
I\omega \ =\ I_{M_1}^{\ N_1} ...  I_{M_n}^{\ N_n}    \,  \omega_{N_1 \ldots N_n} \, dx^{M_1} \wedge \cdots \wedge
dx^{M_n} \, .
  \ee
When acting on the form of the type $(p,q)$, it multiplies $\omega$ by the factor $i^{q-p}$.

Finally, on top of the usual exterior derivative $d$, we introduce the operator
$$d_I \ =\ [d, \iota] \, .$$
 Representing $d$ as the sum of the holomorphic and antiholomorphic (with respect to $I$) exterior derivatives, $d
= \partial_I + \bar \partial_I$ and using \p{iotapq}, we easily derive $d_I = i(\bar \partial_I - \partial_I)$ (and hence $\partial_I = (d + i d_I)/2$).
For an integrable $I$, complex coordinates can be chosen such the complex structure matrix \p{Icompl} is constant.  In this case, we can
write a simple explicit expression for $d_I$,
 \be
\lb{defdI}
d_I \omega \ =\ I_M^{\ S}  \,   \,  \partial_S \omega_{N_1 \ldots N_n} \, dx^M \wedge dx^{N_1} \wedge \cdots
\wedge dx^{N_n} \, .
  \ee

We prove now some simple lemmas.

\vspace{0.1cm}

{\bf Proposition 1.}
\be
\lb{d=IdI}
d_I \omega \ =\ (-1)^n I d(I\omega) \, ,
 \ee
where $n$ is the order of the form.

{\it Proof:} Choose the complex coordinates. Consider the R.H.S. of \p{d=IdI}
and use the complex expression  \p{Icompl} for $I$.
The components $I_M^{\ N}$ are thus constant and the partial derivatives do not act upon them.
The form $d(I\omega)$ has the order $n+1$ and, according to \p{defIomega}, the expression
$I d(I\omega)$ has altogether $(n+1) + n = 2n+1$ factors of $I$.
 This involves $n$ pairs giving $I^2 = -1$ [ this compensates
the factor $(-1)^n$] and we are left with just one unpaired factor.
 We obtain
the expression \p{defdI}.
 In contrast to \p{defdI}, the R.H.S. of \p{d=IdI} has a tensorial form and is valid with {\it any}
choice of coordinates.

\vspace{0.1cm}

{\bf Proposition 2.} The form \p{omegahol}
has the type $(2,0)$ with respect to $I$.

{\it Proof:}
Indeed, using the definition \p{defiota} and the properties \p{quaternion}, it is easy to derive $\iota \omega
= 2i\omega$.

\vspace{0.1cm}

{\bf Proposition 3.} For any complex manifold,
 \be
\lb{dIOmI}
d_I \Omega_I \ =\ \frac 13 \, C_{MPQ} \, dx^M \wedge dx^P \wedge dx^Q \, .
 \ee

{\it Proof:}
Choosing complex coordinates and bearing in mind \p{Icompl},
\p{defdI} and \p{Ccompl}, we derive
$$ d_I \Omega_I =  C_{m p {\bar q}} dz^m \wedge dz^p \wedge d \bar z^{\bar q}
 +  C_{{\bar m} {\bar  p} q} d \bar z^{\bar m} \wedge d\bar z^{\bar p}  \wedge dz^q \, ,$$
which coincides with \p{dIOmI}.

\vspace{0.1cm}

{\bf Corollary:} For the HKT manifolds where the Bismut torsions for $I,J,K$ coincide,
 \be
\lb{dIOI=dJOJ}
 d_I \Omega_I = \ d_J \Omega_J =\ d_K \Omega_K \, .
 \ee

\vspace{0.1cm}

{\bf Proposition 4.}
Let $I,J,K$ be quaternion complex structures. Then
\be\lb{IOm=Om}
\begin{array}{c}
I\Omega_I = \Omega_I\,, \qquad J\Omega_J = \Omega_J\,, \qquad K\Omega_K = \Omega_K\,, \\[6pt]
J\Omega_I = K\Omega_I = -\Omega_I\,, \qquad
I\Omega_J = K\Omega_J = -\Omega_J\,, \qquad
I\Omega_K = J\Omega_K = -\Omega_K \, .
\end{array}
\ee

{\it Proof:} Let us prove the relation $J\Omega_I = -\Omega_I$.
By definition,
$$ J\Omega_I \ =\ J_M^{\ P} J_S^{\ Q} I_{PQ} \, dx^M \wedge dx^S \, .$$
On the other hand,
$$J_M^{\ P} J_S^{\ Q} I_{PQ} \ =\ -K_{MQ} J_S^{\ Q} = -I_{MS} \, . $$
Other relations are proved similarly.
\vspace{0.1cm}

{\bf Remark.}  The condition \p{dIOI=dJOJ} can be rewritten bearing in mind
\p{d=IdI} and the first line in \p{IOm=Om} as
 \be
\lb{IdOmI-JdOmJ}
 I d\Omega_I = J d \Omega_J = K d \Omega_K \, .
 \ee

We are ready now to prove the main theorem

\begin{thm}
\be
\lb{drondom=0}
 \partial_I (\Omega_J + i\Omega_K) \ =\ 0 \,.
 \ee
\end{thm}

\begin{proof}
The real and imaginary parts of \p{drondom=0} give a kind of Cauchy-Riemann conditions
\be
\lb{Cauchy}
 d\Omega_J - d_I \Omega_K \ =\ 0, \ \ \ \ \ d\Omega_K + d_I \Omega_J = 0 \, .
\ee
Consider the first relation. We obtain
 $$ d_I \Omega_K  \stackrel{1}{=} I d(I\Omega_K)  \stackrel{4}{=} - I d\Omega_K = -JK d\Omega_K
\stackrel{remark}{=} - J^2 d\Omega_J = d\Omega_J \, .$$
The number ``1'' above the equality sign means {\it in virtue of the Proposition 1}, etc.

The relation $d\Omega_K + d_I \Omega_J = 0$ is proved similarly.

\end{proof}

\setcounter{equation}0
\section{Supercharges and reduced supercharges.}

\subsection{ De Rham, K\"ahler - de Rham, Dolbeault, and quasicomplex systems.}

The classical supercharges of the best known de Rham SQM sigma model are
usually presented in the form
\begin{equation} \label{Qtangent}
\begin{array}{rcl}
Q&=&\psi^M\left(P_M-i\Omega_{M,AB}\psi_A\bar\psi_B\right),
\\[8pt]
\bar Q&=&\bar\psi^M\left(P_M-i\Omega_{M,AB}\bar\psi_A\psi_B\right),
\end{array}
\end{equation}
where $A$, $B$ are the tangent space indices, $\psi_A = e_{AM} \psi^M$, $g_{MN}= e_{AM}e_{AN}$, and
\begin{gather}
\Omega_{M,AB}=e_{AN}\big(\partial_M e^N_B+\Gamma^N_{MT}e^T_B\big)
\label{Lorconn}
\end{gather}
are spin connections. The  ``f\/lat'' fermion variables
$\psi_A$, $\bar\psi_A$ constitute, together with $x^M$, $P_M$, the
orthogonal canonically conjugated pairs.

For our purposes,
it is more convenient to express the supercharges in terms of fermionic variables
carrying world indices. The commutation relations are in this case more
complicated,
\begin{equation}
\label{Poisson}
\begin{array}{c}
\big\{x^M, \Pi_N\big\}_{\rm P.B.}\ = \ \delta^M_N \,
, \qquad \big\{\psi^M,\bar\psi^N\big\}_{\rm P.B.}\ = \ -ig^{MN},
\\[8pt]
\big\{\Pi_M,\psi^N\big\}_{\rm P.B.}=-\frac12\,\partial_M g^{NQ}\psi_Q,
\qquad
\big\{\Pi_M,\bar\psi^N\big\}_{\rm P.B.}=-\frac12\,\partial_M g^{NQ}\bar\psi_Q.
\end{array}
\end{equation}
($\{\}_{P.B.}$ stands for a Poisson bracket).
On the other hand, the expressions for the supercharges become much simpler
\cite{quasi},
\begin{equation} \label{QdeRham}
\begin{array}{rcl}
Q&=&\psi^M\left(\Pi_M-\frac i2\, \partial_M g_{NP} \, \psi^N\bar\psi^P\right),
\\[8pt]
\bar Q&=&\bar\psi^M\left(\Pi_M+\frac i2\, \partial_M g_{NP} \,
\psi^P\bar\psi^N\right),
\end{array}
\end{equation}

Note that the momenta $P_M$ and $\Pi_M$ are {\it not} the same. $P_M$ is the variation of the Lagrangian
over $\dot x^M$ while keeping $\psi_A$ and $\bar\psi_A$ fixed. And $\Pi_M$ is the variation of the Lagrangian
over $\dot x^M$ while keeping $\psi^M$ and $\bar\psi^M$ fixed. These two canonical momenta are related as \cite{FIS1}
 \begin{gather}\label{PcherezPi}
P_M \ =\
 \Pi_M+\frac i2\,\Big[(\partial_M e_{AP}) e_{AQ}- (\partial_M e_{AQ}) e_{AP}\Big]\psi^P\bar\psi^Q \, .
\end{gather}

   The {\it covariant} quantum supercharges that act on the wave functions
normalized with the measure  $d\mu = \sqrt{\det (g)} \, d^Nx$
have the same functional form with the operators $\Pi_M = -i\partial/\partial x^M$ and
 $\bar \psi^M = g^{MN} \partial/\partial \psi^N$.

For K\"ahler manifolds, the de Rham complex can be extended to involve an extra pair of supercharges. Being expressed in the same
terms as in \p{QdeRham}, they acquire a very simple form \cite{taming},
\begin{equation} \label{QKahler}
\begin{array}{rcl}
R &=&\psi^N I_N^{\ M} \left(\Pi_M-\frac i2\, \partial_M g_{NP} \, \psi^N\bar\psi^P\right),
\\[8pt]
\bar R &=&\bar\psi^N I_N^{\ M} \left(\Pi_M+\frac i2\, \partial_M g_{NP} \,
\psi^P\bar\psi^N\right),
\end{array}
\end{equation}

Similar simple expressions can be derived for the supercharges of the
Dolbeault complex,
\begin{equation} \label{Qcomplworld}
\begin{array}{rcl}
Q&=& \sqrt{2} \, \psi^m\left(\Pi_m - \frac i2\, \partial_m h_{n \bar p}  \, \psi^n
\bar\psi^{\bar p}\right),
\\[8pt]
\bar Q&=& \sqrt{2} \, \bar\psi^{\bar m}\left(\bar \Pi_{\bar m}  + \frac i2\, \partial_{\bar m}
h_{p \bar n} \,
\psi^p \bar\psi^{\bar n}\right).
\end{array}
\end{equation}

When $h_{m \bar n}$ does not depend on Im($z^p$), one can perform a Hamiltonian
reduction with identification $\Pi_m \equiv \bar \Pi_{\bar m} \to \Pi_M/2$.
\footnote{This implies the convention
$$
z = x+iy, \ \ \bar z = x-iy, \ \ \ \ \frac \partial {\partial z} =
\frac 12 \left(\frac \partial {\partial x} -i \frac \partial {\partial y}
\right), \ \ \frac \partial {\partial \bar z} =
\frac 12 \left(\frac \partial {\partial x} + i \frac \partial {\partial y}
\right)
$$
for each complex coordinate.}
If $h_{m \bar n}$ are real, we obtain the de Rham supercharges
\p{QdeRham}.~\footnote{When deriving this, we have to take into account \p{sym+antisym} and
to bear in mind that the canonical de Rham fermions $\psi^M$
carrying the world indices and satisfying \p{Poisson} have an additional factor
$1/\sqrt{2}$ compared to $\psi^m$.}
 For a generic Hermitian metric \p{sym+antisym},
we obtain the supercharges of a quasicomplex model,
\begin{equation} \label{Qquasi}
\begin{array}{rcl}
Q&=&\psi^M\left[\Pi_M-\frac i2\,\partial_M\left(g_{(NP)}+
i b_{[NP]}\right)\psi^N\bar\psi^P\right],
\\[8pt]
\bar Q&=&\bar\psi^M\left[\Pi_M+\frac i2\,\partial_M\left(g_{(NP)}-i b_{[NP]}\right)\psi^P\bar\psi^N\right] \, .
\end{array}
\end{equation}

\subsection{HKT supercharges}
 The expressions for four real supercharges in an HKT model were derived in \cite{QHKT}.
 They are
   \begin{eqnarray}\label{susy-char1-s-2f}
Q&=& \psi^{M}
\Big(P_{M} -{\textstyle\frac{i}{2}}\,\Omega_{M, AB}\psi^A\psi^B +{\textstyle\frac{i}{12}}\,C_{MNP}\,\psi^{N}\psi^{P}\Big),
\\[6pt]
\label{susy-char1-tr-2f}
Q^{{a=1,2,3}}&=& \psi^Q(I^{{a}})_{Q}{}^{M}\Big(
P_{M} -{\textstyle\frac{i}{2}}\,\Omega_{M,AB}\psi^{A}\psi^{B} - {\textstyle\frac{i}{\,4\,}}\,C_{MNP}\,
\psi^{N}\chi^{P}\Big) \, ,
\end{eqnarray}
where $\Psi^M$ are here {\it real} fermions with
$\{\Psi^M, \Psi^N \}_{\rm P.B.} = -i\delta^{MN}$ and $I^a = \{I,J,K\}$.

We choose now complex coordinates $x^M = \{z^m, {\bar z}^{\bar m}\}$ and construct  the complex combinations

  \begin{equation}\label{SR-def}
S= \frac {Q+iQ^1}{2}\,,\quad
\bar S= \frac{Q-iQ^1}{2}\,,\qquad
R= \frac{Q^2+iQ^3}{2} \,,\quad
\bar R= \frac{Q^2-iQ^3}{2} \, .
\end{equation}
  A short calculation gives

  \begin{equation}\label{S-HKT}
\begin{array}{rcl}
S^{\,{}^{\rm HKT}}&=& \sqrt{2} \, \psi^m\left[P_m-i\Omega_{m,k\bar l}\, \psi^k\bar\psi^{\bar l}\right]\,,\\[7pt]
\bar S^{\,{}^{\rm HKT}}&=& \sqrt{2}\, \bar\psi^{\bar m}\left[\bar P_{\bar m}-i\bar\Omega_{\bar m, k\bar l}\,
\psi^k\bar\psi^{\bar l}\right]\,,
\end{array}
\end{equation}
\begin{equation}\label{R-HKT}
\begin{array}{rcl}
R^{\,{}^{\rm HKT}}&=& \sqrt{2} \, \psi^n{\cal I}_{\,n}{}^{\bar m}
\left[\bar P_{\bar m}-i\left(\bar\Omega_{\bar m,k\bar l}+
{\textstyle\frac12}\,C_{\bar m,k\bar l}\right)
\psi^k\bar\psi^{\bar l}\right]\,,\\[7pt]
\bar R^{\,{}^{\rm HKT}}&=& \sqrt{2} \,
\bar\psi^{\bar n}{\cal I}_{\,\bar n}{}^{m}\left[P_m-i\left(\Omega_{m,k\bar l}+
{\textstyle\frac12}\,C_{mk\bar l}\right) \psi^k\bar\psi^{\bar l}\right]\,,
\end{array}
\end{equation}
where $\Omega_{m, k \bar l} = \Omega_{m, a \bar b} e^a_k { e}^{\bar b}_{\bar l}$ and ${\cal I}=J+iK$.

It is noteworthy that in the expressions for $S$ and $\bar S$, the torsions $C_{mk\bar l}, \,
C_{\bar m k \bar l}$ cancelled
such that $S, \bar S$ represent usual Dolbeault supercharges (cf. (3.15) of Ref.\cite{IS}).
The torsions enter, however in $R$ and $ \bar R$.
For hyper-K\"ahler manifolds, there are no torsions and the expressions \p{R-HKT} simplify.

Substituting the explicit expressions of $\Omega$ and $C$ via vielbeins,
\begin{equation}\label{om-HKT}
\Omega_{m,k\bar l}= e_{\bar l}^{\bar a}\partial_{[m}e_{k]}^a - e_{(m}^{a}\partial_{k)}e_{\bar l}^{\bar a}\,,\qquad
\bar\Omega_{\bar m,k\bar l}= - e_{k}^{a}\bar\partial_{[\bar m}e_{\bar l]}^{\bar a} +
e_{(\bar m}^{\bar a}\partial_{\bar l)}e_{k}^a \,,
\end{equation}
\begin{equation}\label{C-HKT}
{\textstyle\frac{1}{2}}\,C_{mk\bar l}= -e_{\bar l}^{\bar a}\partial_{[m}e_{k]}^a +
e_{[m}^{a}\partial_{k]}e_{\bar l}^{\bar a}\,,\qquad
{\textstyle\frac{1}{2}}\,\bar C_{\bar m\bar lk}= -
e_{k}^{a}\bar\partial_{[\bar m}e_{\bar l]}^{\bar a} + e_{[\bar m}^{\bar a}\partial_{\bar l]}e_{k}^a \,,
\end{equation}
we derive for the supercharges
\begin{equation}\label{S-HKT-c}
\begin{array}{rcl}
S^{\,{}^{\rm HKT}}&=&  \sqrt{2} \, \psi^m\left[P_m-i\left(e_{\bar l}^{\bar a}\partial_{m}e_{k}^a\right)
\psi^k\bar\psi^{\bar l}\right],\\[7pt]
\bar S^{\,{}^{\rm HKT}}&=& \sqrt{2} \, \bar\psi^{\bar m}\left[\bar P_{\bar m}+
i\left(e_{k}^{a}\bar\partial_{\bar m}e_{\bar l}^{\bar a}\right) \psi^k\bar\psi^{\bar l}\right],
\end{array}
\end{equation}
\begin{equation}\label{R-HKT-c}
\begin{array}{rcl}
R^{\,{}^{\rm HKT}}&=& \sqrt{2} \psi^n{\cal I}_{\,n}{}^{\bar m}
\left[\bar P_{\bar m}-i\left(e_{\bar l}^{\bar a}\bar\partial_{\bar m}e_{k}^a\right)
\psi^k\bar\psi^{\bar l}\right],\\[7pt]
\bar R^{\,{}^{\rm HKT}}&=&  \sqrt{2} \bar\psi^{\bar n}{\cal I}_{\,\bar n}{}^{m}
\left[P_m+i\left(e_{k}^{a}\partial_{m}e_{\bar l}^{\bar a}\right) \psi^k\bar\psi^{\bar l}\right].
\end{array}
\end{equation}

At the last step, we go over from the momenta $P_{ m}$, $\bar P_{\bar m}$ to the momenta
$\Pi_{ m}$, $\bar \Pi_{\bar m}$ (which are relevant when $\psi^m$ and $\bar \psi^{\bar m}$ rather than
$\psi^a$ and $\bar \psi^{\bar a}$ are treated as
fundamental dynamic variables) according to
\p{PcherezPi}. The supercharges
 take the simple nice form
\begin{equation}\label{SR-HKT-s}
\begin{array}{rcl}
S^{\,{}^{\rm HKT}}&=& \sqrt{2} \, \psi^m\left[\Pi_m-
{\textstyle\frac{i}{2}}\left(\partial_{m}h_{k\bar l}\right) \psi^k\bar\psi^{\bar l}\right],\\[8pt]
\bar S^{\,{}^{\rm HKT}}&=& \sqrt{2} \, \bar\psi^{\bar m}
\left[\bar \Pi_{\bar m}+{\textstyle\frac{i}{2}}\left(\bar\partial_{\bar m}h_{k\bar l}\right)
\psi^k\bar\psi^{\bar l}\right],
\end{array}
\end{equation}
\begin{equation}\label{SR-HKT-s1}
\begin{array}{rcl}
R^{\,{}^{\rm HKT}}&=& \sqrt{2} \, \psi^n{\cal I}_{\,n}{}^{\bar m}\left[\bar \Pi_{\bar m}-
{\textstyle\frac{i}{2}}\left(\bar\partial_{\bar m}h_{k\bar l}\right)
\psi^k\bar\psi^{\bar l}\right],\\[8pt]
\bar R^{\,{}^{\rm HKT}}&=& \sqrt{2} \, \bar\psi^{\bar n}{\cal I}_{\,\bar n}{}^{m}\left[\Pi_m+
{\textstyle\frac{i}{2}}\left(\partial_{m}h_{k\bar l}\right)\psi^k\bar\psi^{\bar l}\right].
\end{array}
\end{equation}

We observe a remarkable similarity with \p{QdeRham}, \p{QKahler}. For an HKT manifold, the matrix
${\cal I}_n{}^{\bar m}$ plays the same role as the usual complex structure for the K\"ahler - de Rham complex.
${\cal I}$ can thus be called the matrix of {\it hypercomplex structure}. The form ${\cal I}_{mn} \, dz^m \wedge dz^n$
is closed, as dictated by \p{drondom=0}.

When $h_{m \bar n}$ does not depend on the imaginary coordinate parts, one can perform the Hamiltonian reduction.
As an HKT manifold is a complex manifold of a special kind, we obtain after reduction a quasicomplex model
of a special kind. The reduced supercharges are
\begin{equation} \label{SRquasi}
\begin{array}{rcl}
S^{\,{}^{\rm quasi}}&=&\psi^M\left[\Pi_M-{\textstyle\frac{i}{2}}\,\partial_{M}\left(g_{KL}+ib_{KL}\right) \psi^K\bar\psi^{L}
\right],\\[7pt]
\bar S^{\,{}^{\rm quasi}}&=&\bar\psi^{M}\left[\Pi_{M}+{\textstyle\frac{i}{2}}\,\partial_{M}\left(g_{KL}+ib_{KL}\right) \psi^K
\bar\psi^{L}\right],
\end{array}
\end{equation}
\begin{equation} \label{SRquasi1}
\begin{array}{rcl}
R^{\,{}^{\rm quasi}}&=&\psi^N{\cal I}_{\,N}{}^{M}\left[\Pi_{M}-
{\textstyle\frac{i}{2}}\,\partial_{M}\left(g_{KL}+ib_{KL}\right)
\psi^K\bar\psi^{L}\right],\\[7pt]
\bar R^{\,{}^{\rm quasi}}&=&\bar\psi^{N}{\cal I}_{\,N}{}^{M}\left[\Pi_M+
{\textstyle\frac{i}{2}}\,\partial_{M}\left(g_{KL}+ib_{KL}\right)\psi^K \bar\psi^{L}\right]
\end{array}
\end{equation}

When the imaginary part of the metric $b_{KL}$ vanish, the supercharges
 \p{SRquasi}, \p{SRquasi1} boil down
to the K\"ahler supercharges \p{QdeRham}, \p{QKahler}. When it does not,
we are dealing with the {\it K\"ahler quasicomplex} model to be discussed in
more details in the next section.

\section{Hamiltonian reduction and superfields.}
\setcounter{equation}0

Hamiltonians of supersymmetric systems are expressed in components,
and Hamiltonian reduction is usually described in components too ---
see the component expressions for the reduced supercharges \p{Qquasi}, \p{SRquasi},
\p{SRquasi1} in the previous section. But it is interesting and instructive to see
what does it correspond to in Lagrangian superfield formulation.

\subsection{Dolbeault $\to$ quasicomplex de Rham.}

The Dolbeault complex is described by a set of chiral complex
({\bf 2}, {\bf 2}, {\bf 0})
superfields $Z^m$
\cite{IS}. They are expressed into components as
\be
\lb{Zjdecomp}
Z^m \ =\ z^m + \sqrt{2}\, \theta \psi^m - i \theta \bar \theta {\dot z}^m
\ee
The corresponding supersymmetry transformations are
\be
\lb{SUSYtranZ}
\delta z^m = -\sqrt{2}\, \epsilon \psi^m\,, \qquad
\delta \psi^m = i \sqrt{2}\, \bar \epsilon \dot{z}^m\, , \\ [4pt]
\delta \bar z^m = \sqrt{2}\, \bar\epsilon \bar \psi^m\,, \qquad
\delta \bar\psi^m = - i \sqrt{2}\,  \epsilon \dot{\bar z}^m\, .
\ee
We set now $\psi^m = \sqrt{2}\, \chi^m, z^m = x^m +iy^m$ to obtain
\be
\lb{SUSYtranX}
\delta x^m = - \epsilon \chi^m + \bar\epsilon \bar\chi^m \,, \qquad
\delta y^m = i (\epsilon \chi^m + \bar\epsilon \bar\chi^m) \, , \\ [4pt]
\delta \chi^m   =  \bar\epsilon (i \dot{x}^m - \dot{y}^m)\,,  \qquad
\delta \bar \chi^m = -    \epsilon (i\dot{x}^m + \dot{y}^m)\, .
\ee
and observe that \p{SUSYtranZ} {\it coincides} with the transformation law for
a ({\bf 1}, {\bf 2}, {\bf 1}) real superfield
\be
\label{XMdecomp}
X^m \equiv \ X^M = \ x^M+\theta\chi^M+\bar\chi^M\bar\theta + B^M\theta\bar\theta,
\ee
if identifying ${\dot y}^m = {\rm Im} ({\dot z}^m) \equiv B^M$.
\footnote{The observation that the supertransformation laws for the multiplets
with the same net number of the fermionic and bosonic components, but with a different distribution
of the latter among the
dynamic and auxiliary fields, coincide under such identification
was made long time ago in  \cite{BerPash,GatesRana}. This was discussed in the Hamiltonian reduction context
in \cite{BKMO} and in gauging approach (when the Hamiltonian commutes with Im$(\Pi_m)$,
one can impose the first
class constraint Im$(\Pi_m) = 0$ and treat the system as a gauge one)
in \cite{DI06}.}

The Lagrangian of the pure Dolbeault sigma model (without gauge field) is expressed via chiral superfields as
\be
\lb{LDolbeault}
L \ =\ -{\textstyle\frac 14} \int  d\bar\theta d\theta  \, h_{m\bar n} (Z, \bar Z) DZ^m \bar D \bar Z^{\bar n}
\ee
with Hermitian metric $h_{m \bar n}$
and
\begin{equation}\label{der-n2}
D_\theta=\partial_{\theta}-i\bar \theta \partial_t\,,\qquad \bar D_\theta=-\partial_{\bar\theta}+i \theta \partial_t\,.
\end{equation}
 If the metric does not depend on Im($z^m$), one can perform the Hamiltonian
reduction. The reduced  {\it Lagrangian} should be expressed via the superfields $X^M$ as
\be
\label{Lquasicomplex}
L^{\rm reduced}  = \ -\frac12\int d\bar\theta d\theta\, [g_{(MN)}(X) + ib_{[MN]}(X)]  \, DX^M\bar D X^N \, ,
  \ee
where $g/2$ and $b/2$ are real and imaginary parts of the Hermitian metric $h$, according to \p{sym+antisym}.
 Heuristically, \p{Lquasicomplex} is obtained from
\p{LDolbeault} by substituting $Z^m, \bar Z^{\bar m} \to 2X^M$, while taking into account \p{sym+antisym}.
When $h_{m \bar n}$ is real, this is the usual de Rham model. When it involves an antisymmetric imaginary part, we arrive
at the quasicomplex de Rham model of Ref.\cite{quasi}.

The  fact that the reduction of \p{LDolbeault}
gives \p{Lquasicomplex} looks very natural. It can be accurately derived in the following way:
{\it (i)} Express the Lagrangian \p{LDolbeault} into components.
{\it (ii)} It is invariant under the shifts $y^m \to y^m + c^m$
 (the corollary of the fact that the Hamiltonian commutes
 with the corresponding canonical momentum). In other words, it does not explicitly depend on $y^m$,
 but only on
$\dot{y}^m$.
{\it (iii)} Substitute  $\sqrt{2} \chi^M$ for $\psi^m$ and $B^M$ for $\dot{y}^m$.
The result coincides with the component expansion of
\p{Lquasicomplex}.

A more detailed justification of this procedure at the component level is given in Appendix~A.

\subsection{HKT $\to$ quasicomplex K\"ahler.}
Consider now the ${\cal N} = 4$ supersymmetric HKT model. The Lagrangian is expressed via
linear\,\footnote{There are also models expressed via nonlinear multiplets \cite{Delduc}, which we will not discuss here.}
${\cal N} = 4$
supermultiplets of the type ({\bf 4},\,{\bf 4},\,{\bf 0})  \cite{Copa,CKT,IL,IKL}.
A ({\bf 4},\,{\bf 4},\,{\bf 0}) multiplet lives in the ${\cal N} = 4$ superspace with the coordinates
$(t,\theta^{i{k}^\prime})$,
$(\overline{\theta^{i{k}^\prime}})= -\epsilon_{ij}\epsilon_{{k}^\prime{l}^\prime}\theta^{j{l}^\prime}
\equiv -\theta_{k'j} $.
The indices $i=1,2$ and $k^\prime=1,2$ are doublet indices of the ${\rm SU}_L(2)$ and ${\rm SU}_R(2)$ groups respectively, which
form the full automorphism group ${\rm SO}(4)={\rm SU}_L(2)\times{\rm SU}_R(2)$ of
the ${\cal N}{=}\,4$ superalgebra. Each multiplet carries a 4-vector or two spinor indices.
Its component decomposition is
\begin{equation}\label{X}
{\cal X}^{i l'} = x^{i l'}-\theta^{i{k}^\prime}\chi_{{k}^\prime}^{l'} +
i\theta^{i{k}^\prime}\theta_{{k}^\prime k }\dot x^{kl'} -
{\textstyle\frac{i}{3}}\,\theta^{i{i}^\prime} \theta_{{i}^\prime k} \theta^{k{k}^\prime}
\dot\chi_{{k}^\prime}^{l'}
-{\textstyle\frac{1}{12}}\,\theta^{k{k}^\prime}  \theta_{{k}^\prime j}
\theta^{j{i}^\prime}  \theta_{{i}^\prime k} \, \ddot x^{\,k l'},
\end{equation}
and so it encompasses four real bosonic component  fields $(\overline{x^{ik'}})=
-\epsilon_{ij}\epsilon_{k'  l'}x^{jl'}$
and four real fermionic component fields
$(\overline{\chi^{i^\prime k}})=-\epsilon_{i^\prime j^\prime}\epsilon_{k l}\chi^{j^\prime l}$.

The set $x^{i k'}$ satisfying the pseudoreality condition can be represented as 4 real coordinates
\be
    x^M =  {\textstyle\frac 1 {2}}\, (\sigma^M)_{k' i} x^{i k'}\, , \ \ \ \ \ \ \ \ \ \
\sigma^M = (\vec{\sigma}, i)
\ee
 or else  as two complex coordinates $v^{m}$, $m=1,2$,
\footnote{We are changing notation here reserving the symbol z for the complex coordinates of the {\it reduced} model,
see Eq.\p{sootvet} below.}
 \be
\lb{v12}
x^{i k'} \ =\ \left( \begin{array}{rr} \bar v^2 & \bar v^1 \\ v^1 & - v^2 \end{array} \right) \, .
 \ee
The same holds for the superfields ${\cal X}^{i k'}$.

The second representation (via two complex coordinates) is convenient when performing
the Hamiltonian reduction. We may represent $v^m = x^m + iy^m$ and express the laws of supersymmetry
transformations via $x^m$ and $y^m$. Similar to what was the case for the ${\cal N} = 2$ superfields,
one can be convinced that these laws coincide with the supersymmetry transformations for the

({\bf 2},\,{\bf 4},\,{\bf 2}) multiplet if identifying $\dot{y}^m$ with the auxiliary fields $B^M$
(see \cite{DI07} for the discussion of the reduction ({\bf 4},\,{\bf 4},\,{\bf 0})
$\to$ ({\bf 2},\,{\bf 4},\,{\bf 2}) in superfield language
 using gauging procedure).
Thus, to perform the Hamiltonian reduction using the Lagrangian language, one  should only
substitute $ \dot{y}^m  \to B^M$ in the component expression
for the Lagrangian.

A wide class of HKT models are described by the superfield Lagrangian  involving $n$ ({\bf 4},\,{\bf 4},\,{\bf 0})
linear multiplets,
\begin{equation}
\label{act1}
L = \int  d^4 \theta\, {\cal L}({\cal X}_\alpha)\,,
\end{equation}
where $\alpha = 1,\ldots, n$ is the flavor index.

\subsubsection{4-dimensional model.}
Consider as the simplest example the model with only one multiplet, $n=1$.

The simplest HKT metric is a conformally flat metric in 4 dimensions,
\be
\lb{metr4}
ds^2 = G(x)\, dx^M dx^M \, .
\ee
The complex structures can be chosen as
\be
\lb{Icanon}
\!\!
I_M^{\ N} = \left( \begin{array}{cccc} 0&-1&0&0 \\ 1&0&0&0 \\ 0&0&0&-1 \\ 0&0&1&0 \end{array} \right),\,\,
J_M^{\ N} =  \left( \begin{array}{cccc} 0&0&1&0 \\ 0&0&0&-1 \\ -1&0&0&0 \\ 0&1&0&0 \end{array} \right)
 ,\,\,
K_M^{\ N} = \left( \begin{array}{cccc} 0&0&0&1 \\ 0&0&1&0 \\ 0&-1&0&0 \\ -1&0&0&0 \end{array} \right).
\ee
These are self-dual matrices expressible via 't Hooft symbols. The characteristic for a HKT manifold
closed holomorphic form is
\be
\lb{omega4}
\omega = \Omega_J + i \Omega_K = 2\, G(x)\, dv^1 \wedge dv^2 \, ,
\ee
i.e. ${\cal I}_{mn} =  G(x)\, \epsilon_{mn}$.
The complex supercharges \p{SR-HKT-s} with this hypercomplex structure matrix
were written down in \cite{QHKT,taming}.   When $G$ depends only on Re($v^m$), one can perform a Hamiltonian
reduction. After that, we are left with only one complex coordinate
$z = {\rm Re}(v^1) + i  {\rm Re}(v^2)$, the complex metric involves only one
component and is real. We obtain
the usual ${\cal N} = 4$ K\"ahler  model on a manifold of real dimension 2.

Note that in Ref.\cite{quasi} a certain nontrivial
quasicomplex 2-dimensional model was constructed and studied.
It was observed that the spectrum of this model involves degenerate quartets and enjoys ${\cal N} = 4$ supersymmetry.
The corresponding metric  cannot be obtained, however, by a Hamiltonian reduction of an HKT metric and the observed
extended supersymmetry has a different origin.

\subsubsection{4$n$--dimensional model.}

 When $n > 1$, the situation becomes more complicated and more interesting.
 The 8--dimensional model described by two linear  ({\bf 4}, {\bf 4}, {\bf 0})
multiplets was studied in details in \cite{FIS2}.
We consider here a model with an arbitrary number $n$ of such multiplets. Anticipating a subsequent reduction,
it is convenient to use the complex notation and describe each of them as the  ${\cal N}{=}\,4$ superfield
${\cal V}^{m}_{\alpha} (t; \, \theta,\eta,\bar\theta,\bar\eta)$
 where $\theta,\eta,\bar\theta,\bar\eta$
are odd coordinates of ${\cal N}{=}\,4$ superspace and $m = 1,2$.
 The fields ${\cal V}^{m}$ are related to ${\cal X}^{i k'}$ as in  \p{v12}.

To make things quite transparent, we can   represent them in terms of
${\cal N}{=}\,2$ superfields. Each superfield ${\cal V}^m_\alpha$ is expressed via
 two ({\bf 2}, {\bf 2}, {\bf 0}) complex  superfields $V^m_\alpha$ and their conjugates.
Its expansion of  in $\eta, \bar\eta$ reads  \cite{FIS2}
\begin{eqnarray}\label{V-4-2exp}
&{\cal V}^{m}_\alpha = V^{m}_\alpha +\eta\, \epsilon^{mn}\bar D
\bar V^{\bar n}_\alpha -i\eta\bar\eta \,\dot{V}^{m}_\alpha \,,\quad &
\bar{\cal V}^{\bar{m}}_\alpha  = \bar V^{\bar{m}}_\alpha -\bar \eta\, \epsilon^{{\bar m}\bar{n}}  D V^{n}_\alpha
+i\eta\bar\eta \,\dot{\bar V}^{\bar{m}}_\alpha   \,,
\end{eqnarray}
with $D \equiv D_\theta$ and $\bar D \equiv \bar D_\theta$ defined in \p{der-n2}.
$V^{m}_\alpha(\theta,\bar\theta)$
are chiral superfields, $\bar D_\theta V^{m}_\alpha = 0$, $D_\theta\bar V^{\bar{m}}_\alpha = 0$.  They have a standard component expansion
\begin{eqnarray}\label{V-2-exp}
&V^m_\alpha = v^m_\alpha + \sqrt{2}\, \theta\, \psi^{m}_\alpha -i\theta\bar\theta\, \dot{v}^{m}_\alpha \,,\quad &
\bar V^{\bar m}_\alpha = \bar v^{\bar m}_\alpha - \sqrt{2}\, \bar\theta\, \bar\psi^{\bar m}_\alpha  +i\theta\bar\theta\, \dot{\bar v}^{\bar m}_\alpha \,,
\end{eqnarray}

The superfield action is
\be
\lb{act4}
S={\textstyle\frac14}\displaystyle{\int}dt\, d\theta d\bar\theta\, d\eta d\bar\eta\, {\cal L}({\cal V}_\alpha, \bar{\cal V}_\alpha) = \
{\textstyle\frac14} \int dt \, d\theta d\bar\theta \,
\left( \Delta_{m\bar n}^{\alpha\beta} {\cal L} \right) D V^m_\alpha \bar D \bar V^{\bar n}_\beta \, ,
\ee
where
\begin{equation} \lb{Del2}
\Delta_{m\bar n}^{\alpha\beta} {\cal L} \equiv  \frac{\partial^2 {\cal L}(V, \bar V)}{\partial V^m_\alpha \partial \bar V^{\bar n}_\beta }   +
\epsilon_{m k} \epsilon_{\bar n \bar l} \frac{\partial^2 {\cal L}(V, \bar V)}{\partial \bar V^{\bar k}_\alpha \partial V^l_\beta}\,.
\end{equation}
Note that, for $\alpha = \beta$,
\begin{equation} \lb{id-Del2}
\Delta^{\alpha\alpha}_{m\bar n}
=\delta_{m\bar n}
\frac {\partial^2 {\cal L}}{\partial V_\alpha^k \partial \bar V_\alpha^{\bar k}} \qquad\qquad \mbox{(here no summation with respect to $\alpha$)}\, .
\end{equation}

The R.H.S. of Eq.\p{act4} expresses the action in terms of the ${\cal N} =2$
superfields. It is invariant under "hidden" supersymmetry transformations \cite{Delduc}
 \be
\lb{N2trans}
\delta_\eta V^m_\alpha = -\epsilon_\eta \epsilon^{mn} \bar D \bar V^{\bar n}_\alpha, \ \ \ \ \ \ \
\delta_\eta \bar V^{\bar m}_\alpha = \bar \epsilon_\eta
\epsilon^{\bar m \bar n}  D V^{ n}_\alpha \, .
 \ee

Integrating it further over $d^2\theta$, we obtain
 the component Lagrangian. Its bosonic part
\footnote{See Appendix B for the complete expression.}  reads
\be
\lb{lagr4-b}
L_b \ =\ \left( \Delta_{m\bar n}^{\alpha\beta} {\cal L} \right) \dot v^m_\alpha \dot {\bar v}^{\bar n}_\beta \, ,
\ee

The Lagrangian \eqref{lagr4-b}
implies the target space metric
   \begin{equation} \lb{metr8}
ds^2 =
\left( \Delta_{m\bar n}^{\alpha\beta} {\cal L} \right)
dv^m_\alpha d\bar v^{\bar n}_\beta \, .
\end{equation}
The closed holomorphic form is
\be
\lb{omega8}
\Omega =
{\textstyle\frac12} \,
\epsilon_{mk}\,\frac {\partial^2 {\cal L}}{\partial \bar v^{\bar k}_\alpha \partial v^n_\beta}\,
dv^m_\alpha \wedge dv^n_\beta\, .
\ee

\subsubsection{The reduced model and its superfield description.}

When $L$ is real, the metric \p{metr8} is Hermitian, but not necessarily symmetric and real.
To be able to perform the Hamiltonian reduction, we have to impose an extra
constraint for the metric to be independent on Im($v^m_\alpha$).
 This implies also certain constraints on ${\cal L}$.
A generic admissible form of ${\cal L}$ will be written and discussed in Appendix C. Here we write a restricted
Ansatz for ${\cal L}$ generating the metric with a constant antisymmetric under $\alpha \leftrightarrow \beta$ part,
\be
\lb{constantAns}
{\cal L}\,\, =  {\cal K} \, - \ {\textstyle\frac i2}  \,
{\cal C}_{\alpha\beta}
 \left( {\cal V}^1_\alpha \bar {\cal V}^1_\beta - {\cal V}^2_\alpha \bar {\cal V}^2_\beta  \right)
\ee
with a real function ${\cal K}$ generating the
real symmetric part of the target space metric \eqref{metr8} that does not
depend on Im($v^m_\alpha$),
and a real
constant antisymmetric ${\cal C}_{\alpha\beta}=-{\cal C}_{\beta\alpha}$
(the coefficients are chosen for further convenience).

Consider the second term in \p{constantAns}.
Its contribution to the bosonic kinetic Lagrangian reads
    \be
\label{antikinbos}
L_{\rm bos} \ =\ - 2\, {\cal C}_{\alpha\beta}
\left(\dot{x}^1_\alpha \dot{y}^1_\beta -
\dot{x}^2_\alpha \dot{y}^2_\beta \right) \, .
    \ee
 Bearing in mind our recipe $\dot{y}^m_{\alpha} \to B^M_{\alpha} $, this gives
\be
\label{antikinbos_xB}
L_{\rm bos}^{\rm red} \ =\  -2\, {\cal C}_{\alpha\beta}
\left( \dot{x}^1_\alpha B_\beta^1 - \dot{x}^2_\alpha B^M_\beta \right)
\ee
in the reduced Lagrangian. The presence of the structure \p{antikinbos_xB}
is characteristic to quasicomplex de Rham models --- see Eq.\p{offL}.
In our case, we are dealing with an ${\cal N} = 4$ model, a  deformation of the
K\"ahler - de Rham complex not studied before.
 To reveal K\"ahlerian nature of the reduced model, we introduce {\it new} complex coordinates
(made out of the real parts of $v^m_\alpha$ and of $B^M_\alpha$),
\be
\lb{sootvet}
z^\alpha = x^1_\alpha + i x^2_\alpha\,, \qquad A^\alpha = B^1_\alpha + i B^2_\alpha \, .
\ee
In this variables, the full bosonic Lagrangian of the reduced model has the form
  \be
\lb{lagr242-b}
 L_{b} = h_{\alpha {\bar \beta}} (z, \bar z) \left(\dot z^{\alpha} \dot {\bar z}^{\bar\beta}  + A^{\alpha}
{\bar A}^{\bar\beta} \right)
 -  {\cal C}_{[\alpha\beta]} \left( \dot{z}^{\alpha} A^{\beta} +
  \dot{\bar z}^{\bar\alpha}\bar
A^{\bar\beta} \right)\, ,
\ee
where
\be
h_{\alpha\bar\beta} \ =\ \frac {\partial^2 {\cal K}}{\partial z^\alpha \partial \bar z^{\bar \beta}} \,
 \ee

Besides the familiar first term with the K\"ahler metric, the Lagrangian
\p{lagr242-b}
involves also an extra term involving $C_{[\alpha \beta]}$. In modifies the {\it full} complex metric
obtained after excluding
the the auxiliary fields $A^\alpha$,
$\bar A^{\bar \beta}$,
 \be
\tilde{h}_{\alpha \bar\beta} \ =\ h_{\alpha \bar\beta} + C_{[\alpha \delta]} h^{\bar \gamma \delta}
C_{[\bar \gamma \bar \beta]} \, ,
  \ee
where
$h^{\bar\alpha \beta} h_{\beta \bar\gamma} = \delta^{\bar\alpha}_{\bar\gamma},
 h_{\alpha \bar\gamma}h^{\bar\gamma \beta} = \delta^{\beta}_{\alpha}$.

We will see now how this system is expressed
in terms of the ({\bf 2}, {\bf 4}, {\bf 2}) superfields.
In contrast to the $n=1$ case where the extra terms in \p{lagr242-b} were absent, and the reduced model is a
well-known K\"ahler --
de Rham model, with the Lagrangian representing
a superspace integral of the K\"ahler potential, when $n \geq 2$, the superfield Lagrangian is somewhat more
complicated.

We consider a set of $n$ chiral ({\bf 2}, {\bf 4}, {\bf 2}) multiplets
 ${\cal Z}^\alpha(t; \, \theta,\bar\theta,\eta,\bar\eta)$
satisfying the constraints
\begin{equation}\label{242-const}
\bar D_\theta\,{\cal Z}^\alpha=0\,,\qquad \bar D_\eta\,{\cal Z}^\alpha=0\,,
\end{equation}
where $D_\theta$, $\bar D_\theta$ are defined by \eqref{der-n2} and
\begin{equation}\label{der-n2-eta}
D_\eta=\partial_{\eta}-i\bar \eta \partial_t\,,\qquad \bar D_\eta =
-\partial_{\bar\eta}+i \eta \partial_t\,.
\end{equation}

It is convenient to represent ${\cal Z}^\alpha$  via ${\cal N} = 2$ superfields \cite{IS}:
the usual chiral ({\bf 2}, {\bf 2}, {\bf 0}) superfield
$Z^\alpha(\theta,\bar\theta)$ and the superfield $\Phi^\alpha(\theta,\bar\theta)$ of the type
({\bf 0}, {\bf 2}, {\bf 2}),
\be
\lb{calZalpha}
{\cal Z}^\alpha \ =\ Z^\alpha + \sqrt{2}\, \eta \, \Phi^\alpha - i \eta \bar \eta\, \dot{Z}^\alpha \, ,
\ee
where
\be
\lb{Zalpha}
Z^\alpha= z^\alpha + \sqrt{2}\, \theta\,\phi^\alpha -i\theta\bar\theta\,\dot z^\alpha\,, \qquad \bar D_\theta\,{Z}^\alpha=0\,,
\ee
\be
\lb{Phialpha}
\Phi^\alpha =  \varphi^\alpha + \sqrt{2}\, \theta\,A^\alpha -i\theta\bar\theta\,
\dot \varphi^\alpha\,,\qquad \bar D_\theta\,{\Phi}^\alpha=0\,.
\ee
In \eqref{Zalpha}, \eqref{Phialpha} the dynamical fields $z^\alpha$ and complex auxiliary fields $A^\alpha$ are bosonic
whereas $\phi^\alpha$, $\varphi^\alpha$ are fermionic.

Now, the standard K\"ahler model is described by the action
$\sim {\int}dt\, d\theta d\bar\theta\, d\eta d\bar\eta\, {\cal K}({\cal Z},\bar{\cal Z})$.
We note that one can add to this expression  $F$-terms of a certain particular form,
\footnote{This is specific for $d=1$. In $d \geq 2$ field theories, it is absent. Probably, this is the reason why such a
structure
was not considered before.}
\be
\lb{act242}
 S =  {\textstyle\frac14}\displaystyle{\int}dt\, d\theta d\bar\theta\,
d\eta d\bar\eta\, {\cal K}({\cal Z},\bar{\cal Z})
+  {\textstyle\frac14}  \displaystyle{\int}dt\, d\theta \, d\eta \,
{\cal F}_{\alpha\beta}({\cal Z}){\cal Z}^{\alpha}\dot{\cal Z}^{\beta}-  {\textstyle\frac14}
\displaystyle{\int}dt\, d\bar\theta\, d\bar\eta\,
\bar{\cal F}_{\bar\alpha\bar\beta}(\bar{\cal Z})\bar{\cal Z}^{\bar\alpha}\dot{\bar{\cal Z}}^{\bar\beta}\,.
\ee
The action is expressed in terms of the ${\cal N} =2$ superfields \p{Zalpha}, \p{Phialpha}  as follows,
\be
\lb{act242N2}
 S =  - {\textstyle\frac14}\displaystyle{\int}dt\, d\bar\theta d\theta \,
 h_{\alpha \bar \beta}(Z, \bar Z) \left( D Z^\alpha \bar D \bar Z^{\bar \beta} - 2
\Phi^\alpha \bar \Phi^{\bar \beta} \right)
+  {\textstyle\frac 1{\sqrt{2}}}  \left[ \displaystyle{\int}dt\, d\theta \,
{\cal C}_{[\alpha\beta]}( Z) \, \Phi^{\alpha}\dot{ Z}^{\beta} \ +\ {\rm c.c.} \right] \, \,
\ee
where
\be
\lb{C-def}
{\cal C}_{\alpha\beta} =
\ {\cal F}_{\alpha\beta} + Z^\gamma \partial_\alpha {\cal F}_{\gamma\beta}\,.
\ee
The full component expression of this Lagrangian is written in Appendix B. Its bosonic part \p{lagr242-b-add}
depends on the holomorphic antisymmetric tensor $C_{[\alpha\beta]}$. The expression
\p{lagr242-b} corresponds to the particular choice
of constant real
${\cal F}_{\alpha\beta} = {\cal C}_{\alpha\beta}$.

Alternatively, one can express the Lagrangian of this deformed K\"ahler  model via an even number
of real ({\bf 1}, {\bf 2}, {\bf 1}) multiplets [obviously, ({\bf 2}, {\bf 4}, {\bf 2}) =
({\bf 2}, {\bf 2}, {\bf 0}) +
({\bf 0}, {\bf 2}, {\bf 2}) = ({\bf 1}, {\bf 2}, {\bf 1}) + ({\bf 1}, {\bf 2}, {\bf 1})].
The model represents then a particular
case of the quasicomplex de Rham model \p{Lquasicomplex}, with
 the metric $g_{(MN)}$ and the real antisymmetric tensor $b_{[MN]}$ having a particular form
depending on the Hermitian
$h_{\alpha \bar \beta}$ and the holomorphic $C_{[\alpha\beta]}$ in Eq.\p{lagr242-b}.

\section{Conclusions.}
 We list again here the most essential original observations made in this paper.
\begin{enumerate}
\item
We derived the new simple representation \p{SR-HKT-s} for the HKT supercharges. In contrast to \cite{QHKT}, the
supercharges are expressed via complex coordinates and the fermion variables with world (rather than the tangent space)
indices. The second pair of the supercharges involves the holomorphic matrix  ${\cal I}_{mn}$ of {\it hypercomplex structure}.

\item We presented the new quasicomplex K\"ahler -- de Rham model \p{act242} where, in addition to the standard
K\"ahler structure,  the Lagrangian involves  extra $F$-terms of a certain particular form.

\item We have shown that the models of this kind are obtained after a Hamiltonian reduction of  HKT models.
We discussed and justified the known recipe, according to which the canonical velocities corresponding to the variables
subject to reduction  in the original Lagrangian should be replaced by the auxiliary fields, $\dot y^m \to B^M$. In the beginning
of Sect. 4, we showed how the supertransformation laws of the original multiplet and the reduced multiplet match. In Appendix A,
we explored how the Hamiltonian reduction works at the Lagrangian component level for a wide class of systems (not necessarily
supersymmetric.

It is interesting to see how this procedure works for CKT and OKT models. What kind of models are obtained as a result of their
Hamiltonian reduction ?
This question is under study now.

\end{enumerate}

\bigskip
\section*{Acknowledgements}

\noindent

We are indebted to Evgeny Ivanov for illuminating discussions.

S.F. acknowledges support from the RFBR
grants 12-02-00517, 13-02-90430 and
a grant of the IN2P3-JINR Program. He would like to thank SUBATECH,
Universit\'{e} de Nantes, for the warm hospitality
in the course of this study.

\renewcommand\theequation{A.\arabic{equation}} \setcounter{equation}0
\section*{Appendix A: Hamiltonian reduction in \\ \hspace*{4cm} component Lagrangian language.}

As was discussed in the main text, the reduced component Lagrangian is obtained from the original Lagrangian by trading
time derivatives of the coordinates subject to reduction (in our case, time derivatives of the imaginary coordinates parts)
by auxiliary fields.
We will illustrate here how it works by an explicit calculation. Namely, we compare the reduced Hamiltonians obtained by
{\it (i)} Hamiltonian reduction from the original one and {\it (ii)} by the Legendre transformation from the reduced
Lagrangian and show that they coincide.

Our starting point is the complex sigma model with the coordinates
\begin{equation}\label{compl-dyn}
{z}^m=x^m+iy^m\,, \qquad \bar{z}^{\bar m}=x^m-iy^m\, .
\end{equation}
 The metric tensor $h_{m \bar n}$ is Hermitian, but not necessarily real,
\begin{equation}\label{herm-tens}
{h}_{m\bar n}= \frac 12 \left({g}_{(mn)}+i{b}_{[mn]} \right)\,.
\end{equation}
In the case when the real tensors ${g}_{(mn)}$ and ${b}_{[mn]}$  and other structures in the Hamiltonian
do not depend on the imaginary parts $y^m$, we can perform
the Hamiltonian reduction. Disregard for simplicity the fermion  variables (the recipe $\dot{y}^m \to B^M$ works actually not
only for SQM where it is justified by comparing the
supertransformation laws before and after reduction,
but also for purely bosonic systems) and consider the Lagrangian
\begin{equation}\label{lagr-1-app}
L  = {h}_{m\bar n}\dot{z}^m \dot{\bar{z}}^{\bar n}+G_{m}\dot{z}^m
+\bar{G}_{\bar m}\dot{\bar{z}}^{\bar m}-V\,,
\end{equation}
where $G_{m}$, $\bar{G}_{\bar m}$ and $V$ do not depend on the imaginary parts $y^m$.
\footnote{We need not be concerned with their nature, though one can also note that, in the Dolbeault
model we are mostly interested in here (HKT models represent their particular case),
$G_m$  is associated with the gauge potential. In the full Lagrangian that also includes
 fermions, $G_m$ contains in addition
a  bilinear in fermions term.  }

The corresponding Hamiltonian  is
\begin{equation}\label{ham-1-app}
H = (\bar{\pi}_{\bar n}-\bar{G}_{\bar n})(h^{-1})^{\bar n m}(\pi_m -G_{m}) +
 V\,.
\end{equation}
We represent now
\begin{equation}\label{mom-1-app}
\pi_m={\textstyle\frac12}\,\left(p^{(x)}_m-ip^{(y)}_m \right),\qquad
\bar{\pi}_{\bar m}={\textstyle\frac12}\,\left(p^{(x)}_m+ip^{(y)}_m \right).
\end{equation}
and perform the reduction. The reduced Hamiltonian is
\be
\lb{Hreduced}
H^{\rm red} \ =\ {\textstyle\frac14}\, (h^{-1})^{NM} (p_N - 2\bar{G}_N)(p_M - 2G_M) + V \, ,
 \ee
where $  (h^{-1})^{NM} \equiv  (h^{-1})^{\bar n m}$. We are now using capital Latin indices
 and  are not displaying anymore the superscript $^{(x)}$ for
$p$.

On the other hand, the reduced {\it Lagrangian} is obtained from \p{lagr-1-app} by substituting
$\dot{y}^m \to B^M$ ($B^M$ being the real auxiliary fields) and reads
  \be
 \lb{LredwithA}
L^{\rm red}\ =\ h_{MN} (\dot{x}^M + iB^M) (\dot{x}^N - iB^N) + G_M (\dot{x}^M + iB^M) +
\bar{G}_M (\dot{x}^M - iB^M) -V \, .
   \ee
 Representing
\begin{equation}\label{F-1-app}
G_M={\textstyle\frac12}\,\left(R_M+iM_M \right),\qquad
\bar{G}_{ M}={\textstyle\frac12}\,\left(R_M-iM_M \right).
\end{equation}
and  excluding $B^M$, we obtain
   \begin{equation}
  \label{Lreduced}
L^{\rm red} \ =\ \frac 12 \, {G}_{MN} \dot{x}^M \dot{x}^{N} +\left[R_M+b_{MK}(g^{-1})^{KN} M_{N}\right]\dot{x}^M -
{\textstyle\frac12}\,(g^{-1})^{MN}M_{M}M_{N}
-V\,,
   \end{equation}
where
\begin{equation}\label{G-app}
G_{MN}=g_{MN}+b_{MK}(g^{-1})^{KL}b_{LN}\,.
\end{equation}

Bearing in mind that the tensor $(h^{-1})^{NM}$ entering \p{Hreduced}
is expressed as
   \begin{equation}\label{h-1-app}
(h^{-1})^{NM}= 2 \left[ (G^{-1})^{NM} - i(G^{-1})^{NK} b_{KL} (g^{-1})^{LM} \right] \,,
\end{equation}
it is a straightforward exercise to verify that
\p{Hreduced} and \p{Lreduced} are related to
each other by the standard Legendre transformation.

\renewcommand\theequation{B.\arabic{equation}} \setcounter{equation}0
\section*{Appendix B: Component Lagrangians}

\subsection*{B.1. Multiplets ({\bf 4},\,{\bf 4},\,{\bf 0})}

Superfield action \eqref{act4} of interacting linear ({\bf 4},\,{\bf 4},\,{\bf 0}) multiplets
yields the component Lagrangian
$
L = L_{b}+L_{2f}+L_{4f}
$,
\begin{eqnarray}
\lb{lagr4-b-app}
L_{b} &= &
\left( \Delta^{\alpha\beta}_{m\bar n} {\cal L} \right) \dot v_\alpha^m \dot {\bar v}_\beta^{\bar n}  \, ,
\\ [11pt]
L_{2f} &=& \ \ {\textstyle\frac{i}{2}}\left(\Delta^{\alpha\beta}_{m\bar n}{\cal L} \right)
\left(\psi_\alpha^m\dot{\bar\psi}_\beta^{\bar n}-\dot\psi_\alpha^m\bar\psi_\beta^{\bar n} \right)
\nonumber\\ [11pt]
&& + i\left(\partial^\gamma_k \Delta_{m\bar n}^{\alpha\beta}{\cal L}\right)\dot v_\alpha^m\bar\psi_\beta^{\bar n}\psi_\gamma^k
-i\left(\bar\partial^{\gamma}_{\bar k} \Delta_{m\bar n}^{\alpha\beta}{\cal L}\right)\bar\psi_\gamma^{\bar k}\psi_\alpha^m \dot {\bar v}_\beta^{\bar n}
\lb{lagr4-2f-app}\\ [11pt]
&& +{\textstyle\frac{i}{2}}\left(\dot v_\gamma^k\partial^{\gamma}_k
-\dot{\bar v}_\gamma^{\bar k}\bar\partial^{\gamma}_{\bar k}\right)
\left(\Delta^{\alpha\beta}_{m\bar n}{\cal L} \right)\psi_\alpha^m\bar\psi_\beta^{\bar n} \,,
\nonumber\\ [11pt]
\lb{lagr4-4f-app}
L_{4f} &=& {\textstyle\frac{1}{4}}\,\epsilon_{mp}\,\epsilon_{\bar k \bar r}
\left(\Delta^{\gamma\delta}_{r\bar l}\Delta^{\beta\alpha}_{ n\bar p} {\cal L}
\right)\psi_\alpha^{m}\psi_\beta^{n}\bar\psi_\gamma^{\bar k}\bar\psi_\delta^{\bar l}
\, ,
\end{eqnarray}
where
$
\partial^{\,\alpha}_m= {\partial}/{\partial v_\alpha^m}$,
$\bar\partial^{\,\alpha}_{\bar m}= {\partial}/{\partial \bar v_\alpha^{\bar m}}
$
and $\Delta^{\alpha\beta}_{m\bar n} {\cal L}$ is defined by \eqref{Del2}.

\subsection*{B.2. Multiplets ({\bf 2},\,{\bf 4},\,{\bf 2})}

The component Lagrangian $\tilde L = \tilde L_{b}+\tilde L_{2f}+\tilde L_{4f}$ of the superfield action \eqref{act242}
of interacting linear chiral ({\bf 2},\,{\bf 4},\,{\bf 2}) multiplets  has the following form

\be
\lb{lagr242-b-add}
\begin{array}{rcl}
\tilde L_{b} &=& \left(\partial_{\alpha}\bar\partial_{\bar\beta}{\cal K} \right)\left(\dot z^{\alpha} \dot {\bar z}^{\bar\beta} \ +\
A^{\alpha} {\bar A}^{\bar\beta}\right)
-{\cal C}_{[\alpha\beta]}\,\dot{z}^{\alpha} A^{\beta}-
\bar{\cal C}_{[\bar\alpha\bar\beta]}\,\dot{\bar z}^{\bar\alpha}\bar A^{\bar\beta}\, ,
\end{array}
\ee
\be
\lb{lagr242-2f-add}
\begin{array}{rcl}
\tilde L_{2f} &=& \ \ \frac{i}{2}\left(\partial_{\alpha}\bar\partial_{\bar\beta}{\cal K} \right)
\left(\bar\phi^{\bar \beta}\dot\phi^\alpha- \dot{\bar\phi}^{\bar \beta}\phi^\alpha+
\bar\varphi^{\bar \beta}\dot\varphi^\alpha- \dot{\bar\varphi}^{\bar \beta}\varphi^\alpha\right)
\\ [11pt]
&& -{\textstyle\frac12}\,
{\cal C}_{[\alpha\beta]}\left(\varphi^{\alpha}\dot\phi^{\beta}-\phi^{\alpha}\dot\varphi^{\beta} \right)+
{\textstyle\frac12}\,
\bar{\cal C}_{[\bar\alpha\bar\beta]}
\left(\bar\varphi^{\bar\alpha}\dot{\bar\phi}^{\bar\beta}-\bar\phi^{\bar\alpha}\dot{\bar\varphi}^{\bar\beta} \right)
\\ [11pt]
&& - \frac{i}{2}\left[\left(\partial_{\alpha}\partial_{\beta}\bar\partial_{\bar\gamma}{\cal K} \right)\dot z^\alpha
-\left(\bar\partial_{\bar\alpha}\partial_{\beta}\bar\partial_{\bar\gamma}{\cal K} \right)\dot {\bar z}^{\bar\alpha}\right]
\left(\phi^\beta\bar\phi^{\bar \gamma}+ \varphi^\beta\bar\varphi^{\bar \gamma}\right)
\\ [11pt]
&& -\left(\partial_{\alpha}\bar\partial_{\bar\beta}\bar\partial_{\bar\gamma}{\cal K}\right)
A^\alpha\bar\phi^{\bar \beta}\bar\varphi^{\bar \gamma}
+ \left(\bar\partial_{\bar\alpha}\partial_{\beta}\partial_{\gamma}{\cal K}\right)
{\bar A}^{\bar \alpha}\phi^{\beta}\varphi^\gamma
\\ [11pt]
&& -{\textstyle\frac12}\,\phi^{\gamma}\left(\partial_{\gamma}{\cal C}_{[\alpha\beta]} \right)
\dot z^{\alpha}\varphi^{\beta}
+ {\textstyle\frac12}\,\varphi^{\gamma}\left(\partial_{\gamma}{\cal C}_{[\alpha\beta]} \right)
\dot z^{\alpha}\phi^{\beta}
\\ [11pt]
&&
+{\textstyle\frac12}\,\bar\phi^{\bar\gamma}\left(\bar\partial_{\bar\gamma}\bar{\cal C}_{[\bar\alpha\bar\beta]} \right)
\dot {\bar z}^{\bar\alpha}\bar\varphi^{\bar\beta}
- {\textstyle\frac12}\,\bar\varphi^{\bar\gamma}\left(\bar\partial_{\bar\gamma}\bar{\cal C}_{[\bar\alpha\bar\beta]} \right)
\dot {\bar z}^{\bar\alpha}\bar\phi^{\bar\beta}
\,,
\end{array}
\ee
\be
\lb{lagr242-4f-add}
\tilde L_{4f} = -\left(\partial_{\alpha}\partial_{\beta}\bar\partial_{\bar\gamma}\bar\partial_{\bar\delta}{\cal K} \right)
\phi^\alpha\varphi^\beta\bar\phi^{\bar\gamma}\bar\varphi^{\bar\delta}
\, ,
\ee
where
$\partial_\alpha\equiv {\partial}/\partial z^\alpha$,
$\bar\partial_{\bar \alpha}\equiv {\partial}/\partial \bar z^{\bar \alpha}$ and
${\cal C}_{\alpha\beta}=\partial_{\alpha}(z^{\gamma}{\cal F}_{\gamma\beta})$,
$\bar{\cal C}_{\bar\alpha\bar\beta}=
\partial_{\bar\alpha}(\bar z^{\bar\gamma}\bar{\cal F}_{\bar\gamma\bar\beta})$.
Note the identities
\be
\lb{id-C}
\partial_{\alpha}{\cal C}_{[\beta\gamma]}+
\partial_{\beta}{\cal C}_{[\gamma\alpha]}+\partial_{\gamma}{\cal C}_{[\alpha\beta]} = 0\,,\qquad
\bar\partial_{\bar\alpha}\bar{\cal C}_{[\bar\beta\bar\gamma]}+
\bar\partial_{\bar\beta}\bar{\cal C}_{[\bar\gamma\bar\alpha]}+\bar\partial_{\bar\gamma}\bar{\cal C}_{[\bar\alpha\bar\beta]} = 0\, .
\ee

It is instructive to compare the Lagrangian \p{lagr242-b-add} - \p{lagr242-4f-add} with the Lagrangian of a {\it generic}
quasicomplex ${\cal N} = 2$ model derived in \cite{quasi}. The  expression of the Lagrangian \p{Lquasicomplex} into the components
of  \p{XMdecomp} reads
  \begin{gather}
L=\frac{1}{2}\, g_{(MN)}\big(\dot{x}^M\dot{x}^N+ B^MB^N\big)+b_{[MN]}\dot{x}^M B^N+\frac i2\, g_{(MN)}
\big(\bar\chi^N\nabla{\chi}^M-\nabla\bar\chi^N{\chi}^M\big)
\nonumber
\\
\hphantom{L=}{}
-\frac{1}{2} \, b_{[MN]}\big(\bar\chi^N\dot{\chi}^M-\dot{\bar\chi}^N{\chi}^M \big)
-\frac12\, \partial_P\partial_Q\big(g_{(MN)}+i b_{[MN]}\big)\chi^M\bar\chi^N\chi^P\bar\chi^Q
\nonumber
\\
\hphantom{L=}{}
+G_{M,PQ} B^M\chi^P \bar\chi^Q-\frac{1}{2}\left(\partial_M b_{[NP]}+\partial_N b_{[MP]}\right)\dot{x}
^P\chi^M\bar\chi{}^N\, ,
\label{offL}
\end{gather}

\begin{gather}
G_{M,PQ}=\Gamma_{M,PQ}-\frac i2\left(\partial_M b_{[PQ]}+\partial_P b_{[QM]}+\partial_Q b_{[MP]}\right) \, ,
\label{Chris}
\end{gather}
 with $\Gamma_{M,
PQ}$ being the standard Christoffels for $g_{(MN)}$,
\be
\Gamma_{M,PQ}=\frac12\big[\partial_P g_{(MQ)}+\partial_Q g_{(MP)}-\partial_M g_{(PQ)}\big]\, ,
\ee
and
 \be
\nabla\psi^M=\dot{\psi}^M+\Gamma^M_{NQ}\dot{x}{}^N\psi^Q \, .
 \ee

One can observe that the Lagrangian \p{offL} involves among other terms the $b$-dependent
4-fermion term $\sim b \chi \bar \chi \chi \bar \chi$ and the
terms $ (\partial b) F \chi \bar \chi$, which do not have a counterpart in \p{lagr242-2f-add}, \p{lagr242-4f-add}.
Well, one can explicitly show that, in the ${\cal N} = 4$ case for {\it particular} $b_{[MN]}$
depending on holomorphic $C_{\alpha\beta}$, these contributions vanish, indeed.

\renewcommand\theequation{C.\arabic{equation}} \setcounter{equation}0
\section*{Appendix C: Reduction of general HKT models}

 We will construct here a generic form of the HKT prepotential in \p{act1} allowing reduction and show that the bosonic action
of the reduced model coincides with \p{lagr242-b} with generic ${\cal F}_{\alpha\beta}$.

We define the superfields
\be
{\cal V}^m_\alpha \ =\ {\cal X}^m_\alpha + i {\cal Y}^m_\alpha, \ \ \ \ \ \ {\cal Z}^\alpha \ =\ {\cal X}^1_\alpha + i{\cal X}^2_\alpha , \ \ \ \ \ \ \Xi^\alpha \ =\
{\cal Y}^1_\alpha + i{\cal Y}^2_\alpha
  \ee
with bosonic component fields
 \be
v^m_\alpha \ =\ x^m_\alpha + i y^m_\alpha, \ \ \ \ \ \ z^\alpha \ =\ x^1_\alpha + ix^2_\alpha , \ \ \ \ \ \ \xi^\alpha \ =\
y^1_\alpha + iy^2_\alpha \, .
  \ee

 We consider now the operator  $ \Delta_{m\bar n}^{\alpha\beta}$ entering \p{Del2} and express it in terms  of $z, \bar z, \xi, \bar \xi$,
 \be
\lb{Delta-via-zxi}
      2 \Delta_{1 \bar 1}^{\alpha\beta} \ =\ \frac {\partial^2}{\partial z^\alpha \partial \bar z^\beta} +
\frac {\partial^2}{\partial z^\beta \partial \bar z^\alpha} + i \left(
\frac {\partial^2}{\partial z^\alpha \partial  \xi^\beta} - \frac {\partial^2}{\partial z^\beta \partial  \xi^\alpha}
+ \frac {\partial^2}{\partial \bar z^\alpha \partial \bar \xi^\beta} - \frac {\partial^2}{\partial \bar z^\beta \partial  \bar \xi^\alpha}
\right) \nn
 + \frac {\partial^2}{\partial \xi^\alpha \partial \bar \xi^\beta} +
\frac {\partial^2}{\partial \xi^\beta \partial \bar \xi^\alpha} \nn
     2 \Delta_{2 \bar 2}^{\alpha\beta} \ =\ \frac {\partial^2}{\partial z^\alpha \partial \bar z^\beta} +
\frac {\partial^2}{\partial z^\beta \partial \bar z^\alpha} - i \left(
\frac {\partial^2}{\partial z^\alpha \partial  \xi^\beta} - \frac {\partial^2}{\partial z^\beta \partial  \xi^\alpha}
+ \frac {\partial^2}{\partial \bar z^\alpha \partial  \bar \xi^\beta} - \frac {\partial^2}{\partial \bar z^\beta \partial  \bar \xi^\alpha}
\right)  \nn
+ \frac {\partial^2}{\partial \xi^\alpha \partial \bar \xi^\beta} +
\frac {\partial^2}{\partial \xi^\beta \partial \bar \xi^\alpha} \nn
      2 \Delta_{1 \bar 2}^{\alpha\beta} \ =\ i \left( \frac {\partial^2}{\partial \bar z^\alpha \partial  z^\beta} -
\frac {\partial^2}{\partial z^\alpha \partial \bar z^\beta} \right) +
\frac {\partial^2}{\partial \xi^\alpha \partial  z^\beta} - \frac {\partial^2}{\partial z^\alpha \partial  \xi^\beta}
+ \frac {\partial^2}{\partial \bar z^\alpha \partial  \bar \xi^\beta} - \frac {\partial^2}{\partial \bar z^\beta \partial  \bar \xi^\alpha}
 \nn
+ i \left( \frac {\partial^2}{\partial \bar \xi^\alpha \partial  \xi^\beta} -
\frac {\partial^2}{\partial \xi^\alpha \partial  \bar \xi^\beta} \right) \nn
     2 \Delta_{2 \bar 1}^{\alpha\beta} \ =\ 
i \left( \frac {\partial^2}{\partial z^\alpha \partial \bar z^\beta} -
\frac {\partial^2}{\partial \bar z^\alpha \partial  z^\beta} \right) +
\frac {\partial^2}{\partial \xi^\alpha \partial  z^\beta} - \frac {\partial^2}{\partial z^\alpha \partial  \xi^\beta}
+ \frac {\partial^2}{\partial \bar z^\alpha \partial  \bar \xi^\beta} - \frac {\partial^2}{\partial \bar z^\beta \partial \bar \xi^\alpha}
 \nn
- i \left( \frac {\partial^2}{\partial \bar \xi^\alpha \partial  \xi^\beta} -
\frac {\partial^2}{\partial \xi^\alpha \partial  \bar \xi^\beta} \right) \, .
   \ee

The second term in \p{constantAns} is expressed via ${\cal Z}, \bar {\cal Z}, \Xi, \bar \Xi$ as
  \be
\lb{const-via-zxi}
-{\textstyle\frac 12} \, {\cal C}_{\alpha\beta} \left({\cal Z}^\alpha \Xi^\beta + \bar {\cal Z}^\alpha \bar \Xi^\beta \right) .
 \ee
It is linear in $\xi, \bar \xi$, but the result of the action of \p{Delta-via-zxi}
on  \p{const-via-zxi} gives a constant not depending on the imaginary parts of $v^m_\alpha$ entering
$\xi, \bar \xi$.

We  generalize now \p{const-via-zxi} by introducing the following term in the prepotential
  \be
\lb{gen-via-zxi}
\Delta{\cal L} = - {\textstyle\frac 12} \, \Big[ {\cal F}_{\alpha\beta}({\cal Z}) \, {\cal Z}^\alpha \Xi^\beta +
  \bar {\cal F}_{\alpha\beta}(\bar {\cal Z})\, \bar {\cal Z}^\alpha \bar \Xi^\beta \Big] \ .
 \ee
It is not difficult to observe that only the mixed terms in \p{Delta-via-zxi} involving both $z$ and $\xi$ derivatives give a nonzero result
when acting on \p{gen-via-zxi}. The result does not depend on $\xi, \bar \xi$ and is expressed in the form \p{lagr242-b}. {\it Q.E.D.}

\end{document}